\documentclass[letterpaper, 10pt, conference]{ieeeconf}
\IEEEoverridecommandlockouts

\let\labelindent\relax

\usepackage{amsmath,amsfonts}
\usepackage{algorithmic}
\usepackage{algorithm}
\usepackage{array}
\usepackage[caption=false,font=normalsize,labelfont=sf,textfont=sf]{subfig}
\usepackage{textcomp}
\usepackage{stfloats}
\usepackage{url}
\usepackage{verbatim}
\usepackage{graphicx}
\usepackage{cite}
\usepackage{booktabs}
\usepackage{cite}
\usepackage{amsthm}
\usepackage{algorithmic}
\usepackage{graphicx}
\usepackage{textcomp}
\usepackage{xcolor}
\usepackage{graphicx}
\usepackage{indentfirst}
\usepackage{CJK}
\usepackage{graphicx}
\usepackage{epstopdf}
\usepackage{flushend}
\usepackage{balance}
\usepackage{lettrine}
\usepackage{flushend,cuted}
\usepackage{blindtext}
\usepackage[caption=false,font=normalsize,labelfont=sf,textfont=sf]{subfig}
\usepackage{amsmath,amsfonts}
\usepackage{algorithm}
\usepackage{array}
\usepackage{enumerate}
\usepackage{enumitem} 

\usepackage{amsmath,amsfonts}
\usepackage{algorithmic}
\usepackage{array}
\usepackage{textcomp}
\usepackage{stfloats}
\usepackage{url}
\usepackage{verbatim}
\usepackage{balance}
\usepackage{xcolor}
\usepackage[caption=false,font=normalsize,labelfont=sf,textfont=sf]{subfig}
\def\BibTeX{{\rm B\kern-.05em{\sc i\kern-.025em b}\kern-.08em
		T\kern-.1667em\lower.7ex\hbox{E}\kern-.125emX}}
\usepackage{orcidlink}
\usepackage{makecell}
\usepackage{multirow}
\usepackage{booktabs} 
\usepackage{cite}
\usepackage{amsmath,amssymb,amsfonts}
\usepackage{algorithmic}
\usepackage{graphicx}
\usepackage{textcomp}
\usepackage{xcolor}
\usepackage{algorithmic}
\usepackage{algorithm}
\usepackage{epstopdf}    
\usepackage{graphicx}

\usepackage{orcidlink}

\newtheorem{lemma}{\bf Lemma}
\newtheorem{assumption}{\bf Assumption}

\newtheorem{ccase}{Case}
\newtheorem{theorem}{\bf Theorem}

\begin{document}
	
	\title{\LARGE \bf Effective Fixed-Time Control for Constrained Nonlinear System}
	
	\author{Chenglin Gong, Ziming Wang\orcidlink{0000-0001-7000-9578}, Guanxuan Jiang\orcidlink{0009-0001-3686-6266}, Xin Wang\orcidlink{0000-0003-2070-960X} and Yiding Ji\orcidlink{0000-0003-2678-7051}$^{\dag}$
		\thanks{$^{\dag}$Corresponding author.}
		\thanks{Chenglin Gong and Xin Wang are with the College of Electronic and Information Engineering, Southwest University, 400700, Chongqing, China.}
		\thanks{Ziming Wang and Yiding Ji are with the Systems Hub, Robotics and Autonomous Systems Thrust, The Hong Kong University of Science and Technology (Guangzhou), 511458, Guangzhou, China.}
            \thanks{Guanxuan Jiang is with the Information Hub, Computational Media and Arts Thrust, The Hong Kong University of Science and Technology (Guangzhou), 511458, Guangzhou, China.}
        \thanks{Finantial support: National Natural Science Foundation of China grants 62303389 and 62373289; Guangdong Basic and Applied Research Funding grants 2022A151511076 and 2024A1515012586; Guangdong Research Platform and Project Scheme grant 2024KTSCX039; Guangzhou-HKUST(GZ) Joint Funding grants
        2024A03J0618, 2024A03J0680 and 2025A03J3960.}
	}
	
	\maketitle
	
	\begin{abstract}	
        In this paper, we tackle the state transformation problem in non-strict full state-constrained systems by introducing an adaptive fixed-time control method, utilizing a one-to-one asymmetric nonlinear mapping auxiliary system. Additionally, we develop a class of multi-threshold event-triggered control strategies that facilitate autonomous controller updates, substantially reducing communication resource consumption. Notably, the self-triggered strategy distinguishes itself from other strategies by obviating the need for continuous real-time monitoring of the controller's state variables. By accurately forecasting the subsequent activation instance, this strategy significantly optimizes the efficiency of the control system. Moreover, our theoretical analysis demonstrates that the semi-global practical fixed-time stability (SPFTS) criterion guarantees both tracking accuracy and closed-loop stability under state constraints, with convergence time independent of initial conditions. Finally, simulation results reveal that the proposed method significantly decreases the frequency of control command updates while maintaining tracking accuracy.
	\end{abstract}
	
	
	\section{introduction}
        Over recent decades, driven by advances in nonlinear control theory and practical demands, the design of controllers for uncertain nonlinear systems has attracted considerable attention \cite{re1,re3,jiajia4,jiajia1,re30}. The backstepping technique is often utilized as an effective method for controlling nonlinear systems, providing a promising approach to enhance transient performance through designed parameter adjustments. Nonlinearity poses significant challenges in controlling nonlinear systems, imposing limitations on controller design. Radial basis function neural networks (RBFNNs), recognized for their strong approximation abilities and linear parameterization, are frequently used to manage nonlinearities by mapping inputs through fixed nonlinear transformations and linearly combining the results \cite{re24,rbfnns}. Recently, integrating backstepping techniques with neural networks in nonlinear systems has led to numerous notable outcomes.
	
	In practical engineering, constraints are imposed on both system output and states. Thus, state-constrained systems have gained significant attention. \cite{re12} introduced a level control method for nonlinear systems with state constraints, while \cite{re13}  tackled the issue of hard state constraints within nonlinear control systems. Furthermore, \cite{re14} developed adaptive control techniques for uncertain nonlinear systems under full-state constraints, requiring strict adherence. However, practical applications often permit limited dynamic fluctuations. Therefore, this paper focuses on controlling nonlinear systems under non-strict full-state constraints to better meet practical engineering needs.
	
	Control designs often enhance system performance and robustness by optimizing settling time. Research has progressed from implicit Lyapunov theorems \cite{re15} to cooperative control methods \cite{re16}, improving transient performance and disturbance rejection, though often without constraints on stability boundaries. In response, \cite{re17} proposed stricter upper-bound estimations. In practice, high convergence speed is crucial for detection accuracy and interference resilience. Although \cite{re18} effectively used fixed-time control's rapid convergence in engineering systems, it compromised signal stability and disturbance rejection. This study introduces a state-constrained fixed-time control framework to improve system stability.
	
	In real-world applications, systems often face a range of external disturbances, requiring dynamic, real-time adjustments to data transmissions to maintain stability, highlighting the need for autonomous decision-making in controllers. Traditional frameworks use fixed-period or pre-defined triggers. Event-triggered control (ETC) principles were established in \cite{re4}, while \cite{re5} demonstrated stability control via nonlinear feedback without triggering. \cite{re6} introduced an event-driven framework addressing the common use of periodic sampling or time-triggered systems in engineering. The self-triggered mechanism in \cite{re7} advanced by predicting future sampling points and updating actuator timing based on current states. As these theories evolved, ETC strategies were refined, as discussed in \cite{re8,jia1,jia2,re9,re10,re11}. This paper offers a framework to systematically analyze performance differences among various ETC strategies.

    \begin{itemize}[leftmargin=*]
		\item An event-triggered adaptive co-design framework that eliminates dependencies on stable controllers and ISS requirements, using a nonlinear mapping technique to convert constrained systems into pure-feedback architectures via logarithmic transformation, addressing initial state effects on convergence time;
        
		\item Development of multi-threshold event-triggered strategies for controller updates, with systematic evaluation and analysis validating their effectiveness in enhancing control accuracy and optimizing communication efficiency;
        
		\item Simulation results demonstrate that the multi-threshold event-triggered strategies effectively optimize resource consumption while maintaining robust tracking performance. This dual advantage underscores the potential of these strategies in enhancing system efficiency and reliability.
	\end{itemize}
	
	
	The remainder of the paper is structured as follows: Section \ref{sec2} describes the system model and RBFNNs. Section \ref{sec3} establishes the constrained system and the ETC strategies. Section \ref{sec4} gives a simulation example. Finally, Section \ref{sec5} concludes the paper.
	
	\section{problem formulation and preliminaries}\label{sec2}
	
	\subsection{System Model}
	We design the non-strict feedback nonlinear systems as:
	\begin{equation}\label{eq1}
		\begin{cases}
			\dot{x} _{i}=h_{i}(\overline{x}_{i+1})\\
			\dot{x} _{n}=g(t)+h_{n}(\overline{x}_{n})\\
			y=x_{n}
		\end{cases}
	\end{equation}
	where $\overline{x}_{n}=[x_{1},x_{2},...,x_{n}]^T\in R^{n}$ denote the system states with $1\leq i\leq n-1$. $h_{i}$ presents the unknown uncertain smooth function. $g(t)$ is the system input and $y(t)\in R$ is the system output. All the states $x_{i}$ satisfy: $-\rho_{s1}<x_{i}<\rho_{s2},s=1,2,...,n$, where $\rho_{s1}$ and $\rho_{s2}$ are positive designed constants, which establish state constraints for this system.	
	
	\subsection{RBFNNs}
	We use radial basis function neural networks (RBFNNs) to approximate the unidentified nonlinear function of the. According to \cite{re19}, for any unknown continuous equation $U(W)$, it has $U(W)=H^{T}\Omega(W)+\varpi(W)$ . $H$ represents the ideal constant weight while $\varpi(W)$ denotes the myopic error, 
	$\Omega(W)=\left [ \Omega_{1}(W),\Omega_{2}(W),...,\Omega_{n}(W)\right ] ^T$ is the basis function vector. By defining $Z_{i}$ as the receptive field center and $D$ as the Gaussian width, we derive $\Omega_{i}(W)=exp(-||W-Z_{i}||/D)$.

	The subsequent lemmas and assumptions are essential for supporting the discussions that follow.
	\begin{lemma}\label{lemma1}
		\cite{re24} For any real variable $x$ and $y$, with the positive constants $r_{1}$,$r_{2}$,$r_{3}$, the following inequality holds
		\begin{eqnarray}
			\lvert x\rvert^{r_{1}}\lvert y\rvert^{r_{2}}\leq\frac{r_{1}r_{2}}{r_{1}+r_{2}}\lvert x\rvert^{r_{1}+r_{2}}+\frac{r_{1}r_{3}^{-\frac{r_{1}}{r_{2}}}}{r_{1}+r_{2}}\lvert y\rvert^{r_{1}+r_{2}}
		\end{eqnarray}
	\end{lemma}
	
	\begin{lemma}\label{lemma2}
		\cite{re24} For a general dynamical system $\dot{x}(t)=l(x(t))$, $x(0)=0$ where the origin is SPFTS with $x\in R_{n}$ and $l():R_{l}\times R^{n}\rightarrow R^{n}$. Design positive parameters $a,b>0$, $I>0$, $p<1$, $q>1$ and $c>0$ such that $\dot{V}(x)=-aV^{q}(x)-bV^{p}(x)+c$ holds.
		\end{lemma}
		
		\begin{lemma}\label{lemma3}
		\cite{re9} For any given constants $\eta_{1}\in R$ and $\eta_{2}>0$, it holds that $0\le \lvert \eta _{1}\rvert-\eta \tanh (\frac{\eta _{1}}{\eta _{2}})\le 0.2785\eta _{2}$.
		\end{lemma}
		
		\begin{assumption}\label{as1}
			\cite{re22} The constant $R^*$, which requires determination, imposes the bound $\lvert \Delta_{i}(w_{i})\rvert\le R^*$ on the term $\Delta_{i}(w_{i})$. The application of Young's inequality to strategically reduce $R^*$ is essential to satisfy the convergence criterion.
		\end{assumption}
		
		\section{Controller Trigger Program}\label{sec3}
		
		\subsection{Controller Design}
		Design the following full-state constraints: 
		\begin{eqnarray}\label{eq7}
			\begin{cases}
				w_{i}=\log_{}{\frac{\rho_{s_{1}}+x_{i}}{\rho_{s_{1}}-x_{i}}}\\
				\dot{w_{i}}=\frac{e^{w_{i}}+e^{-w_{i}}+2}{\rho _{s_{1}}+\rho _{s_{2}}}\dot{x_{i}} 
			\end{cases}
		\end{eqnarray}
		with $i=1,2,\dots,n$. In this system, the logarithmic function serves dual purposes: it simplifies controller design while ensuring strict state boundaries; it also reduces response latency and improves robustness against initial configuration deviations in dynamic interference environments.
		
		Define $\Delta _{i}(w_{i})=(e^{w_{i}}+e^{-w_{i}}+2)/(\rho _{s_{1}}+\rho _{s_{2}})$, where $i=1,2,\dots,n$, to handle \eqref{eq7}. So we get $\dot{w_{i}}=\Delta _{i}(w_{i})\dot{x_{i}}$. Then, the following auxiliary systems have been introduced to facilitate processes in accordance with traditional methods
		\begin{eqnarray}
			\begin{cases}
				L_{i}(\bar{w}_{i+1})=\Delta _{i}(w_{i})h_{i}(\bar{x}_{i+1})-w_{i+1} \\
				L_{n}(\bar{w}_{n})=\Delta_{n}(w_{n})h_{n}(\bar{x}_{n})
			\end{cases}
		\end{eqnarray}
		
		Using the auxiliary system constructed above, we rewrite the constrained system \eqref{eq1} as:
		\begin{eqnarray}\label{eq9}
			\begin{cases}
				\dot{w_{i}}=w_{i+1}+L_{i}(\bar{w}_{i+1})\\
				\dot{w}_{n}=\Delta_{n}(w_{n})g(t)+L_{n}(\bar{w}_{n})
			\end{cases}
		\end{eqnarray}
		
		Define the tracking error $z_i$ as
		\begin{eqnarray}\label{eq10}
			\begin{cases}
				z_{1}=w_{1}-w_{s} \\
				z_{i}=w_{i}-\alpha _{i-1}   ,i=2,3,\dots,n
			\end{cases}
		\end{eqnarray}
		where $w_{s}=\log_{}{\frac{\rho_{s_{1}}+x_{r}}{\rho_{s_{1}}-x_{r}}}$, the expected reference signal is represented as $x_{r}$, and $\alpha _{i}$ denotes the intermediate controller.

        Define $k_{1,o}$, $k_{2,o}$ and $\tau_{o}$ as designed positive constants with $o=1,...,n$. Establish the Lyapunov function $V_1$ and $V_i$, based on the backstepping technique, we get
        \begin{eqnarray}\label{eq111}
			\dot{V}_{1}\le -k_{2,1}z_{1}^{2p}-k_{1,1}z_{1}^{2q}+z_{1}z_{2}+\frac{u_{1}^{2}}{2}+\frac{\lambda _{1}^{2}}{2}+\frac{\tau _{1}}{\varepsilon _{1}}\tilde{\varphi}_{1}\hat{\varphi}_{1}
		\end{eqnarray}
         \vspace{-10pt}
		\begin{eqnarray}\label{eq11}
			\begin{split}
				\dot{V}_{i}\le &-\sum_{o=1}^{i}k_{1,o}z_{o}^{2q} -\sum_{o=1}^{i}k_{2,o}z_{o}^{2p}+\sum_{o=1}^{i}(\frac{u_{o}^{2}}{2} +\frac{\lambda_{o}^{2}}{2})\\&+z_{i}z_{i+1}+\sum_{o=1}^{i}\frac{\tau _{o}}{\varepsilon _{o}}\tilde{\varphi}_{o}\hat{\varphi} _{o}
			\end{split}
		\end{eqnarray}
        Note that the proof of \eqref{eq111} and \eqref{eq11} is shown in APPENDIX.
		
		Then, we construct the final Lyapunov function as
		\begin{eqnarray}\label{eq12}
			V_{n} =V_{n-1} +\frac{1}{2}z_{n}^{2} +\frac{1}{2\varepsilon _{n} } \tilde{\varphi} _{n}^{2}
		\end{eqnarray}
		where $\varepsilon_n =2f_{n}/(2f_{n}-1) $, $f_{n}>\frac{1}{2} $. Then, combining \eqref{eq9}- \eqref{eq11} and calculating the differentiation of \eqref{eq12} yields
		\begin{eqnarray}\label{eq13}
			\begin{split}
				\dot{V} _{n} 
				=&-\sum_{o=1}^{n-1}k_{1,o}z_{o}^{2q} -\sum_{o=1}^{n-1}k_{2,o}z_{o}^{2p}+\sum_{o=1}^{n-1}(\frac{u_{o}^{2}}{2} +\frac{\lambda_{o}^{2}}{2})\\
				&+\sum_{o=1}^{n-1}\frac{\tau _{o}}{\varepsilon _{o}}\tilde{\varphi}_{o}\hat{\varphi} _{o}-\frac{1}{2}z_{n}^{2}-\frac{1}{\varepsilon _{n} } \dot{\hat{\varphi}} _{n}\tilde{\varphi}_{n}\\
				&+z_{n}(-k_{2,n}z_{n}^{2p-1}+\Delta _{n}(w_{n}g(t)) +U_{n}(Z_{n}))   
			\end{split}
		\end{eqnarray}
		
	Then, the time-based controller and adaptive update law are designed as
		\begin{eqnarray}\label{eq16}
			&\alpha _{n}=\frac{1}{\Delta _{n}(w_{n}) }(-k_{1,n} z_{n}^{2q-1}-\frac{1}{2u_{n}^{2}}z_{n} \hat{\varphi } _{n}  \Omega _{n} ^{T} (Z_{n} )\Omega _{n}(Z_{n} ))\\
			\label{eq17}
			&\dot{\hat{\varphi}}_{n}=\frac{\varepsilon _{n}}{2u_{n}^{2}}z_{n}^{2}\Omega _{n}^{T}(Z_{n})\Omega _{n}(Z_{n})-\tau _{n}\hat{\varphi }_{n}  
		\end{eqnarray}
		
		\subsection{Multi-threshold event-triggered strategies}
		This paper proposes four different triggering condition for system \eqref{eq1} to optimize the communication efficiency. Firstly, we define $e(t)=d(t)-g(t)$ is the measurement error between the adaptive controller $d(t)$ and actual system controller $g(t)$, with $g(t)=d(t_{j})$, $\forall t\in [ t_{j},t_{j+1})$.
		
		\itshape {\bf Case 1}: Fixed-threshold strategy
		
		\upshape Under fixed-threshold strategy, $\vartheta$ is a constant. The adaptive controller and triggering condition are designed as
		\begin{eqnarray}\label{eq20}
			&d(t)=\alpha _{n}-\bar{\vartheta}\tanh(\frac{z_{n}\bar{\vartheta}}{\Phi})\\
			\label{eq19}
			&t_{k+1}=\inf \{t \in R| | e(t) \mid \geq \vartheta \},t_{1}=0
		\end{eqnarray}
		where $\Phi$, $\vartheta$ and $\bar{\vartheta}>\vartheta$ are all positive designed parameters. There exists a parameter $\beta(t)$ which, when $\forall t\in [ t_{j},t_{j+1})$, satisfies $\beta(t_{j})=0,\beta_{j+1}=\pm1,\left |\beta(t)\right |\le1$. Combining Lemma\ref{lemma3} and these equations into \eqref{eq13} yields:
		\begin{eqnarray}\label{eq61}
			\begin{split}
				\dot{V}
				\le&-\sum_{o=1}^{n}k_{1,o}z_{o}^{2q}-\sum_{o=1}^{n}k_{2,o}z_{o}^{2p}+\sum_{o=1}^{n}(\frac{u_{o}^{2}}{2}+\frac{\lambda _{o}^{2}}{2})\\&-\sum_{o=1}^{n}\frac{\tau_{o}}{\varepsilon_{o}}\tilde{\varphi}_{o}\hat{\varphi}_{o}+0.2785R\Phi
			\end{split}
		\end{eqnarray}
		
		\itshape {\bf Case 2}: Relative-threshold strategy
		
		\upshape The relative-threshold strategy dynamically modifies triggering thresholds according to the amplitude of the control signal. For larger signals, it utilizes adjustable fault-tolerant thresholds to lengthen update intervals, while for smaller signals, it applies precision-focused thresholds to improve responsiveness. This approach aims to balance stability with performance effectively. The adaptive controller and triggering condition are designed as
		\begin{eqnarray}\label{eq24}
			&d(t)=-(1+\theta)(\alpha_{n}\tanh\frac{z_{n}\alpha_{n}}{\Phi}+\bar{\vartheta}_{1}\tanh\frac{z_{n}\bar{\vartheta}_{1}}{\Phi})\\
			\label{eq23}
			&t_{k+1}=\inf \{t \in R| | e(t) \mid \geq  \theta \left | g(t) \right | +\vartheta _{1} \}
		\end{eqnarray}
		where $\Phi$, $\vartheta_{1}$,$0<\theta<1$ and $\bar{\vartheta}_{1}>\vartheta _{1}/{(1-\theta)}$ are positive designed parameters. There exists continuous time-varying parameter $\beta(t)_{1},\beta(t)_{2}$ which, when $\forall t\in [ t_{j},t_{j+1})$, satisfies $\left|\beta_{1}\right|\le1$ and $\left|\beta_{2}\right|\le1$. Since $\Phi>0$, therefore $a\tanh (a/\Phi)>0$, from \eqref{eq24} we get $d(t)<0$, it follows that $d(t)/(1+\beta _{1}(t)\theta) \le d(t)/(1+\theta)$ and $\left | \beta _{2}(t)\vartheta _{1}/(1+\beta _{1}(t)\theta )\right |\le \vartheta _{1}/(1-\theta )$. Combining Lemma\ref{lemma3} and these equations into \eqref{eq13} yields:
		\begin{eqnarray}\label{eq62}
			\begin{split}
				\dot{V}
				\le&-\sum_{o=1}^{n}k_{1,o}z_{o}^{2q}-\sum_{o=1}^{n}k_{2,o}z_{o}^{2p}+\sum_{o=1}^{n}(\frac{u_{o}^{2}}{2}+\frac{\lambda _{o}^{2}}{2})\\&-\sum_{o=1}^{n}\frac{\tau_{o}}{\varepsilon_{o}}\tilde{\varphi}_{o}\hat{\varphi}_{o}+0.557R\Phi
			\end{split}
		\end{eqnarray}

		\itshape {\bf Case 3}: Switched-threshold strategy
		
		\upshape This switched-threshold approach integrates fixed and relative threshold strategies. The fixed-threshold strategy, $|g(t)|\geq G$, maintains computational errors within predefined limits to ensure stability. Conversely, when $|g(t)|<G$, the approach transitions to a relative threshold, adjusting thresholds proportionally to the signal amplitude for high-precision tracking. By combining these strategies, the triggering condition is designed as follows:
		\begin{eqnarray}\label{eq27}
			t_{k+1}=
			\begin{cases}
				\inf \{t \in R| | e(t) \mid \geq  \theta \left | g(t) \right | +\vartheta _{1} \}&|g(t)|<G\\
				\inf \{t \in R| | e(t) \mid \geq \vartheta \}&|g(t)|\geq G
			\end{cases}
		\end{eqnarray}
		where $G$ is a user-designed parameter and $\theta$, $\vartheta_{1}$ and $\vartheta$ are the same parameters defined before. Then, we have 
		\begin{eqnarray}
			\bar{e}=\sup|e(t)|\le \max \left \{ \theta |g(t)|+\vartheta _{1},\vartheta \right\} ,\forall t\in [ t_{j},t_{j+1})
		\end{eqnarray}
		
		As with fixed-threshold strategy and relative-threshold strategy, define $\left |\beta(t)\right |\le1,\beta(t)_{1}\le1,\beta(t)_{2}\le1$. Similar to \eqref{eq61} and \eqref{eq62}, we obtain:
		\begin{eqnarray}\label{65}
			\begin{split}
				\dot{V}
				\le&-\sum_{o=1}^{n}k_{1,o}z_{o}^{2q}-\sum_{o=1}^{n}k_{2,o}z_{o}^{2p}+\sum_{o=1}^{n}(\frac{u_{o}^{2}}{2}+\frac{\lambda _{o}^{2}}{2})\\&-\sum_{o=1}^{n}\frac{\tau_{o}}{\varepsilon_{o}}\tilde{\varphi}_{o}\hat{\varphi}_{o}+0.8355R\Phi
			\end{split}
		\end{eqnarray}

		\itshape {\bf Case 4}: Self-triggered strategy
		
		\upshape This self-triggered approach calculates the subsequent trigger time $t_{k+1}$ by considering the present control signal $g(t)$, its rate of variation, and dynamic parameters $\theta, \vartheta _{1}, \max  \{ |d(t)|,\pi \}$. By removing the need for constant threshold monitoring typical of traditional event-triggered control, this framework maintains adaptive accuracy. The adaptive controller and triggering condition are designed as
		\begin{eqnarray}\label{eq32}
			&d(t)=-(1+\theta)(\alpha_{n}\tanh\frac{z_{n}\alpha_{n}}{\Phi}+\bar{\vartheta}_{1}\tanh\frac{z_{n}\bar{\vartheta}_{1}}{\Phi})\\
			\label{eq31}
			&t_{j+1}=t_{j}+\frac{\theta |g(t)|+\vartheta _{1}}{\max \left \{ |\dot{d(t)}|,\pi  \right \} } 
		\end{eqnarray}
			where $\Phi$, $\vartheta_{1}$, $0<\theta<1$, $\pi$ and $\bar{\vartheta}_{1}>\vartheta _{1}/{(1-\theta)}$ are positive designed parameters. Similar to the \eqref{65}, one gets
		\begin{eqnarray}
			\begin{split}
				\dot{V}
				\le&-\sum_{o=1}^{n}k_{1,o}z_{o}^{2q}-\sum_{o=1}^{n}k_{2,o}z_{o}^{2p}+\sum_{o=1}^{n}(\frac{u_{o}^{2}}{2}+\frac{\lambda _{o}^{2}}{2})\\&-\sum_{o=1}^{n}\frac{\tau_{o}}{\varepsilon_{o}}\tilde{\varphi}_{o}\hat{\varphi}_{o}+0.557R\Phi
			\end{split}
		\end{eqnarray}
		
		\subsection{Stability Analysis}
		Given the structural similarity of the four aforementioned cases, this paper systematically integrates the final formula to eliminate redundant computations with term $\varLambda$.
		
		\begin{eqnarray}\label{eq34}
			\begin{split}
				\dot{V}\le&-\sum_{o=1}^{n}k_{1,o}z_{o}^{2q}-\sum_{o=1}^{n}k_{2,o}z_{o}^{2p}+\sum_{o=1}^{n}(\frac{u_{o}^{2}}{2}+\frac{\lambda _{o}^{2}}{2})\\&-\sum_{o=1}^{n}\frac{\tau_{o}}{\varepsilon_{o}}\tilde{\varphi}_{o}\hat{\varphi}_{o}+\varLambda
			\end{split}
		\end{eqnarray}
		\begin{eqnarray}
			\varLambda=
			\begin{cases}
				0.2785R\Phi    &\text{Fixed-threshold strategy}\\
				0.557R\Phi    &\text{Relative-threshold strategy}\\
				0.8355R\Phi    &\text{Switched-threshold strategy}\\
				0.557R\Phi    &\text{Self-triggered strategy}\\
			\end{cases}
		\end{eqnarray}
		
		\begin{theorem}
			Consider a closed-loop system consisting of a system \eqref{eq1} with time-varying constraints, a virtual control law, an actual controller \eqref{eq16}, and an adaptive law \eqref{eq17}. The stability of the closed-loop system is then maintained to establish the tracking error performance at a fixed time.
		\end{theorem}
		
		\begin{proof}
			for any $f_{n}>1/2$, we can get $\tau_{o}\tilde{\varphi }_{o}\hat{\varphi }_{o}\le-\frac{\tau_{o}}{\varepsilon_{o}}\tilde{\varphi }_{o}^{2}+\frac{f_{n}\tau_{o}}{2}\tilde{\varphi}_{o}^{2}$, then we get
			\begin{eqnarray}\label{eq36}
				\begin{split}
					\dot{V}\le &-\pi\left(\sum_{o=1}^{n}\frac{z_{o}^{2}}{2}\right)^{q}-\pi\left(\sum_{o=1}^{n}\frac{z_{o}^{2}}{2}\right )^{p}-\pi\left(\sum_{o=1}^{n}\frac{\tilde{\varphi}_{o}^{2}}{2\varepsilon_{o}}\right )^{q}\\&-\pi\left(\sum_{o=1}^{n}\frac{\tilde{\varphi}_{o}^{2}}{2\varepsilon_{o}}\right )^{q}+\sum_{o=1}^{n}\tau_{o}\left(\frac{\tilde{\varphi}_{o}^{2}}{2\varepsilon_{o}}\right )^{q}+\left(\sum_{o=1}^{n}\frac{\tilde{\varphi}_{o}^{2}}{2\varepsilon_{o}}\right )^{p}\\&-\sum_{o=1}^{n}\tau_{o}\left(\frac{\tilde{\varphi}_{o}^{2}}{2\varepsilon_{o}}\right )^{q}-\pi \sum_{o=1}^{n}\frac{\tilde{\varphi}_{o}^{2}}{2\varepsilon_{o}}\\&+\sum_{o=1}^{n}\left ( \frac{u_{o}^{2}}{2}+\frac{\lambda _{o}^{2}}{2}+\frac{f_{o}\tau_{o}}{2}\tilde{\varphi}_{o}^{2}\right )+\Lambda
				\end{split}
			\end{eqnarray}
			where $\pi =\min \left \{2^{q}k_{1,o},2^{p}k_{2,o},\tau _{o},o=1,2,\dots,n \right \} $.
			
			Conjunction Lemma\ref{lemma1}, we can get $\left ( \sum_{o=1}^{n}\frac{\tilde{\varphi }_{o}^{2}}{2\varepsilon _{o}}\right)^{p} \le \sum_{o=1}^{n}\\\frac{\tilde{\varphi }_{o}^{2}}{2\varepsilon _{o}}+\left ( 1-p \right )p^{\frac{p}{1-p}}$. Bringing this inequality to \eqref{eq36} and simplifying it by calculation yields
			\begin{eqnarray}\label{eq37}
				\dot{V}\le -aV_{n}^{q}-bV_{n}^{p}+\sum_{o=1}^{n}\tau _{o}\left ( \frac{\tilde{\varphi }_{o}^{2}}{2\varepsilon _{o}}\right )^{q}-\sum_{o=1}^{n}\frac{\tau _{o}}{2\varepsilon _{o}}\tilde{\varphi }_{o}^{2}+c_{1}
			\end{eqnarray}
			where $a=\pi/(n+1)^{q}$, $b=\pi$, $c_{1}=\sum_{o=1}^{n}(\frac{u_{o}^{2}}{2}+\frac{\lambda_{o}^{2}}{2}+\frac{f_{o}\tau_{o}}{2}\tilde{\varphi}_{o}^{2})+(1-p)p^{\frac{p}{1-p}}+\Lambda $. Assume that there exists an unknown parameter $\gamma_{m}$ satisfying $\left|\tilde{\varphi}_{m}\right|<\gamma _{m}$, then discuss the following two cases:
			\begin{ccase}
				$\gamma _{o} < \sqrt {2\varepsilon _{o}} $: In this case, we have $\sum_{o=1}^{n}\tau _{o} ( \frac{\tilde{\varphi }_{o}^{2}}{2\varepsilon _{o}} )^{q}-\sum_{o=1}^{n}\frac{\tau _{o}}{2\varepsilon _{o}}\tilde{\varphi }_{o}^{2}<0$. Then \eqref{eq37} is rewritten as:
				\begin{eqnarray}
					\dot{V}\le -aV_{n}^{q}-bV_{n}^{p}+c_{1}
				\end{eqnarray}
			\end{ccase}
			\begin{ccase}
				$\gamma _{o} \geq \sqrt {2\varepsilon _{o}}$: In this case, we have $\sum_{o=1}^{n}\tau _{o} ( \frac{\tilde{\varphi }_{o}^{2}}{2\varepsilon _{o}} )^{q}-\sum_{o=1}^{n}\frac{\tau _{o}}{2\varepsilon _{o}}\tilde{\varphi }_{o}^{2}\le \sum_{o=1}^{n}\tau _{o} ( \frac{\tilde{\gamma }_{o}^{2}}{2\varepsilon _{o}} )^{q}-\sum_{o=1}^{n}\frac{\tau _{o}}{2\varepsilon _{o}}\tilde{\gamma }_{o}^{2}$. Then \eqref{eq37} is rewritten as:
				\begin{eqnarray}
					\dot{V}\le -aV_{n}^{q}-bV_{n}^{p}+c_{1}+ \sum_{o=1}^{n}\tau _{o}\left ( \frac{\tilde{\gamma }_{o}^{2}}{2\varepsilon _{o}}\right )^{q}-\sum_{o=1}^{n}\frac{\tau _{o}\tilde{\gamma }_{o}^{2}}{2\varepsilon _{o}}
				\end{eqnarray}
			\end{ccase}
			
			Combining the two cases above, we get:
			\begin{eqnarray}
				\dot{V}\le -aV_{n}^{q}-bV_{n}^{p}+c
			\end{eqnarray}
			where
			\begin{eqnarray}
				c=
				\begin{cases}
					c_{1}, &\text{if $\gamma _{o} < \sqrt {2\varepsilon _{o}}$}\\
					c_{1}+\sum\limits_{o=1}^{n}\tau _{o}\left ( \frac{\tilde{\gamma }_{o}^{2}}{2\varepsilon _{o}}\right )^{q}-\sum\limits_{o=1}^{n}
					\frac{\tau _{o}\tilde{\gamma }_{o}^{2}}{2\varepsilon _{o}}, 
					&\text{if $\gamma _{o} \geq \sqrt {2\varepsilon _{o}}$}
				\end{cases}
			\end{eqnarray}
			
			According to Lemma\ref{lemma2} and \cite{re23}, we can conclude that the signals of the considered closed-loop system are bounded and converge to tight sets $w\in \min \left \{ V(w)\le(\frac{c}{(1-I)a})^{\frac{1}{q}},(\frac{c}{(1-I)b})^{\frac{1}{p}}\right \}$. And the setting time is $T\leq T_{max}:=\frac{1}{aI(q-1)}+\frac{1}{bI(1-p)}$. Then, based on the definition of $\dot{V}$, it can be concluded that the inequality $\left | y-y_{m} \right | \le 2\left ( \frac{c}{(1-I)a}  \right ) ^{\frac{1}{2q} }$ is satisfied. This implies that by selecting suitable parameters, the tracking error can be minimized to a smaller range within a fixed time interval.
		\end{proof}
		
		\section{Illustrative Example}\label{sec4}
		In this section, we present a simulation example to assess the effectiveness of the proposed control algorithm. The nonlinear dynamic system under consideration is shown as
		\begin{eqnarray}
			\begin{cases}
				\dot{x}_{1}=&\cos x_{2}\\
				\dot{x}_{2}=&g+\frac{1}{26}x_{1}^{2}x_{2}^{2}+\frac{1}{5}\sin^{2}(x_{1}x_{2})g\\&+\frac{1}{2}x_{1}^{2}+\frac{1}{26}(u+0.18)^{2}\\
				y=&x_{1}
			\end{cases}
		\end{eqnarray}

		Given the desired trajectory $x_{r}=0.5\sin(0.1t)\cos^{2}(0.6t)$, it is required that component $x_1$ in the system achieves optimal dynamic response performance during reference trajectory tracking. In this paper, to ensure the boundedness of System (1), parameters $\rho_{11}=1,\rho_{12}=2,\rho_{21}=8,\rho_{22}=9$ are configured. The update rate $\dot{\hat{\varphi}}_{i}$ is set to $f_{1}=6,f_{2}=3,\tau_{1}=10,\tau_{2}=10$, with other parameters specified as $u_{1}=1,U_{2}=1,q=1.05,p=1.5$. Initial values are assigned $x_{1}(0)=0.05,x_{2}(0)=0.5$.Subsequently, explanations of four triggering strategies are presented: In the fixed threshold strategy: $k_{1,1}=800,K_{1,2}=19,\bar{\vartheta}=3,\Phi =900,\vartheta=5$; in the relative threshold strategy: $k_{1,1}=800,K_{1,2}=18,\theta =0.1,\bar{\vartheta}_{1}=15,\Phi =900,\vartheta _{1}=0.1$; in the switching threshold strategy: $k_{1,1}=800,K_{1,2}=19,\bar{\vartheta}=\bar{\vartheta}_{1}=10,\Phi =900,\vartheta =4,G=80,\theta =0.1,\vartheta _{1}=2,\Phi =900$; and in the Self-trigger Threshold Strategy: $k_{1,1}=1200,K_{1,2}=3,\theta =0.51,\bar{\vartheta}_{1}=15,\vartheta _{1}=2,\pi=650,\Phi =900$. 

         \begin{figure*}
            \centering
            \subfloat[ ]{\includegraphics[width=0.25\linewidth]{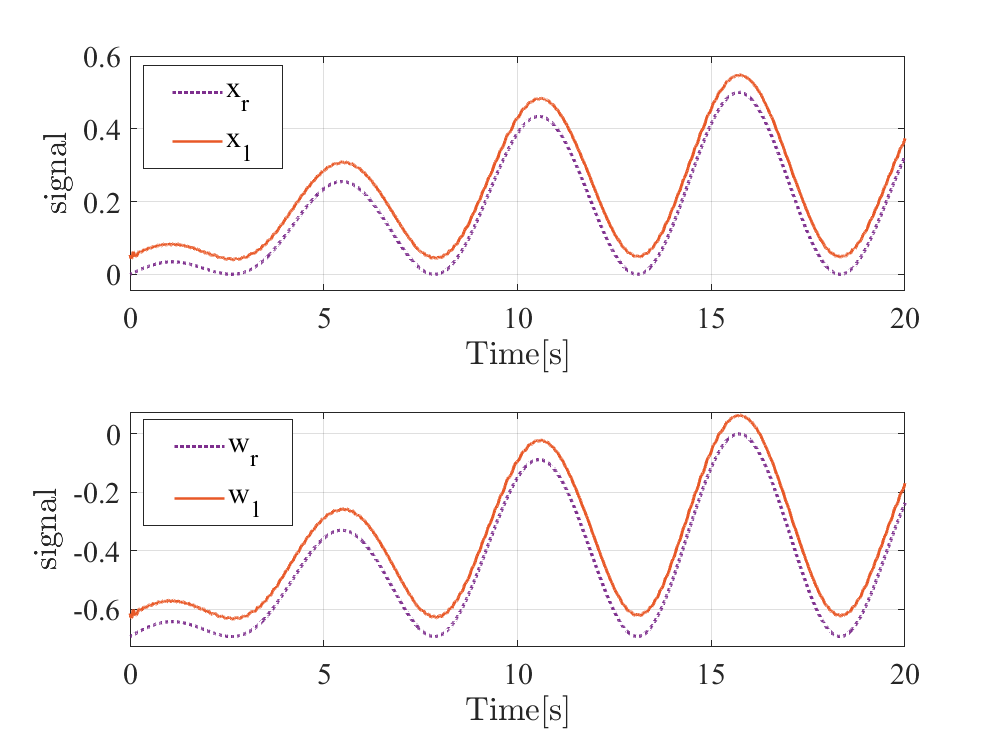}}
            \subfloat[ ]{\includegraphics[width=0.25\linewidth]{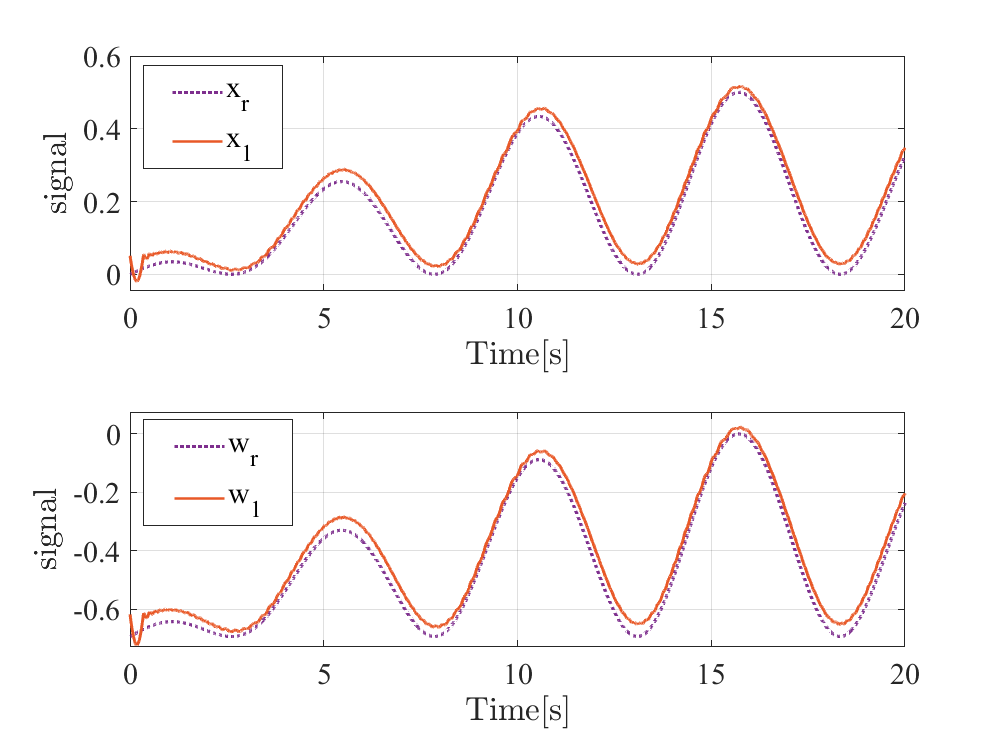}}
            \subfloat[ ]{\includegraphics[width=0.25\linewidth]{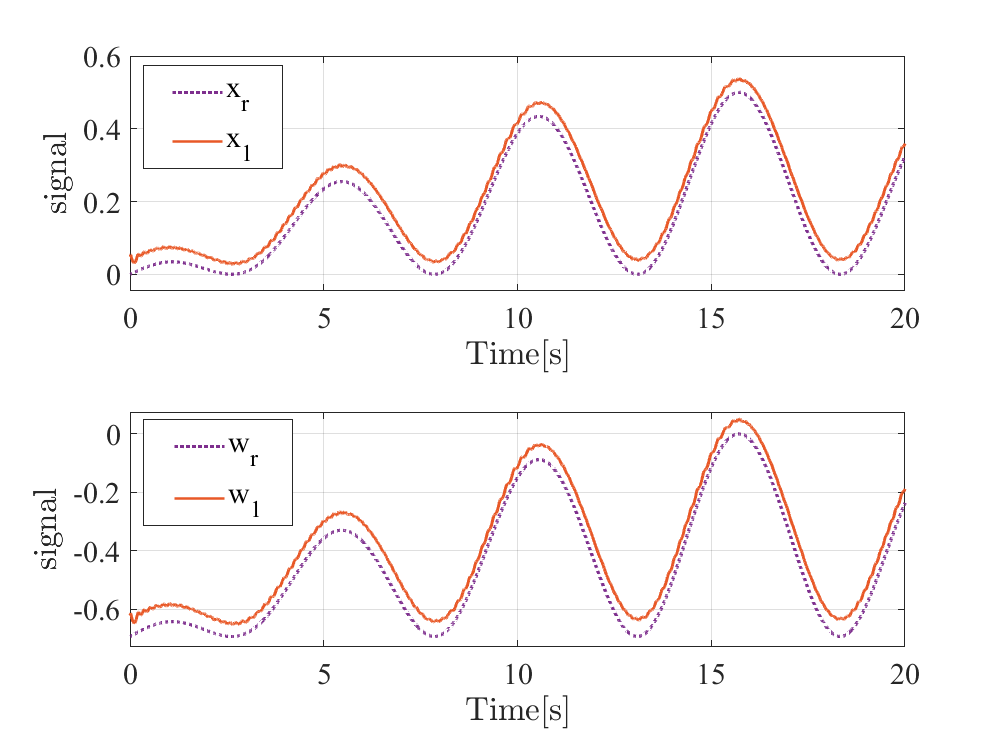}}
            \subfloat[ ]{\includegraphics[width=0.25\linewidth]{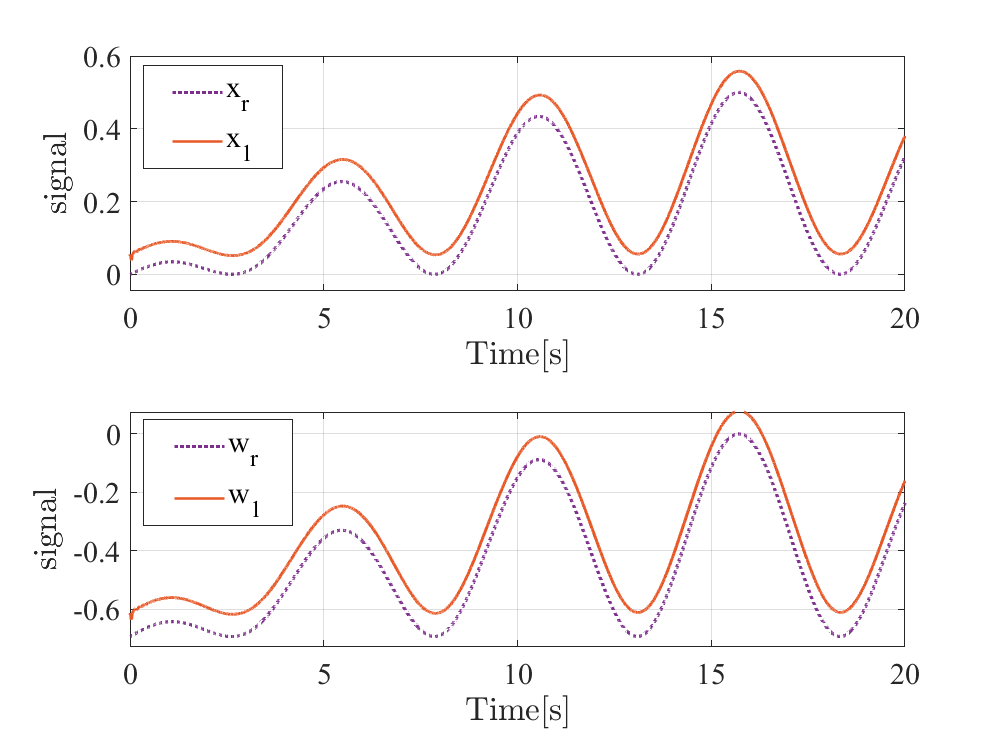}}
            \caption{The consistency tracking performance involving signal $x_{1}$ and restricted signal $w_{1}$, along with their corresponding desired signals $x_{r}$ and $w_{r}$, under different ETC strategies, where (a), (b), (c), and (d) represent the fixed threshold, relative threshold, switched threshold and self-triggered strategies respectively.}
            \label{figure1}
        \end{figure*}

        \begin{figure*}
            \centering
            \subfloat[ ]{\includegraphics[width=0.25\linewidth]{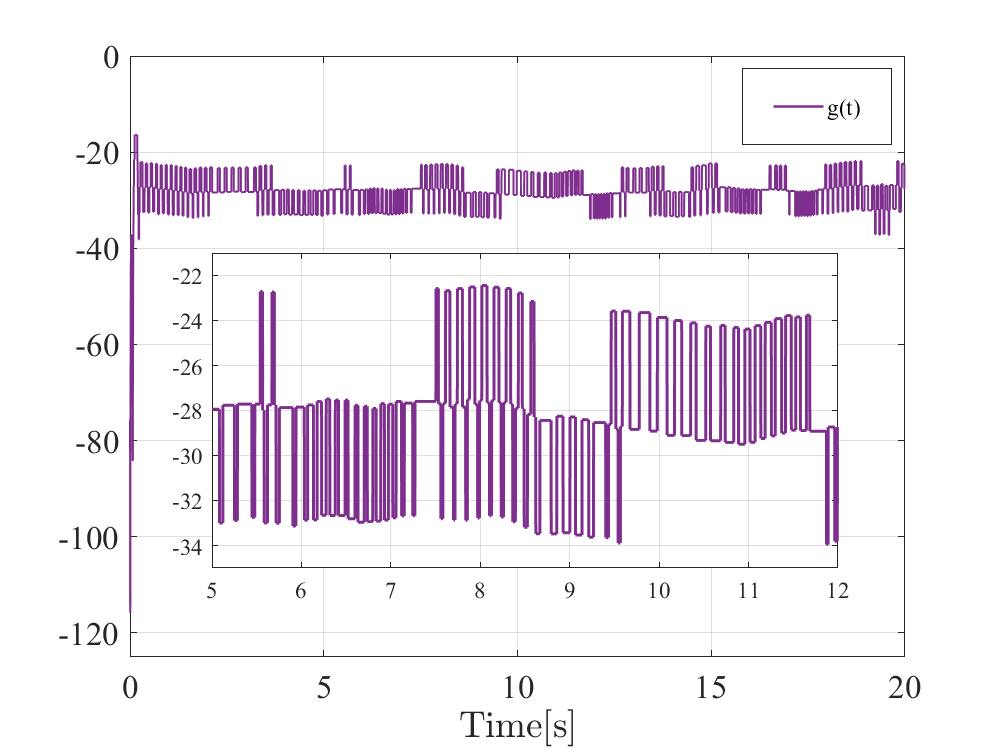}}
            \subfloat[ ]{\includegraphics[width=0.25\linewidth]{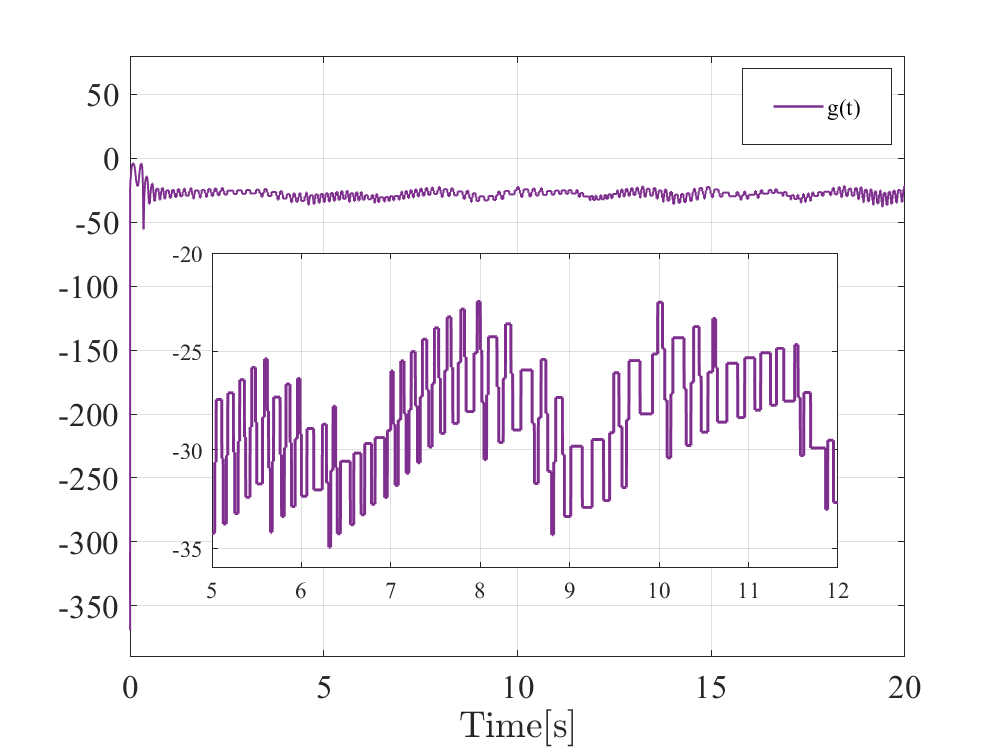}}
            \subfloat[ ]{\includegraphics[width=0.25\linewidth]{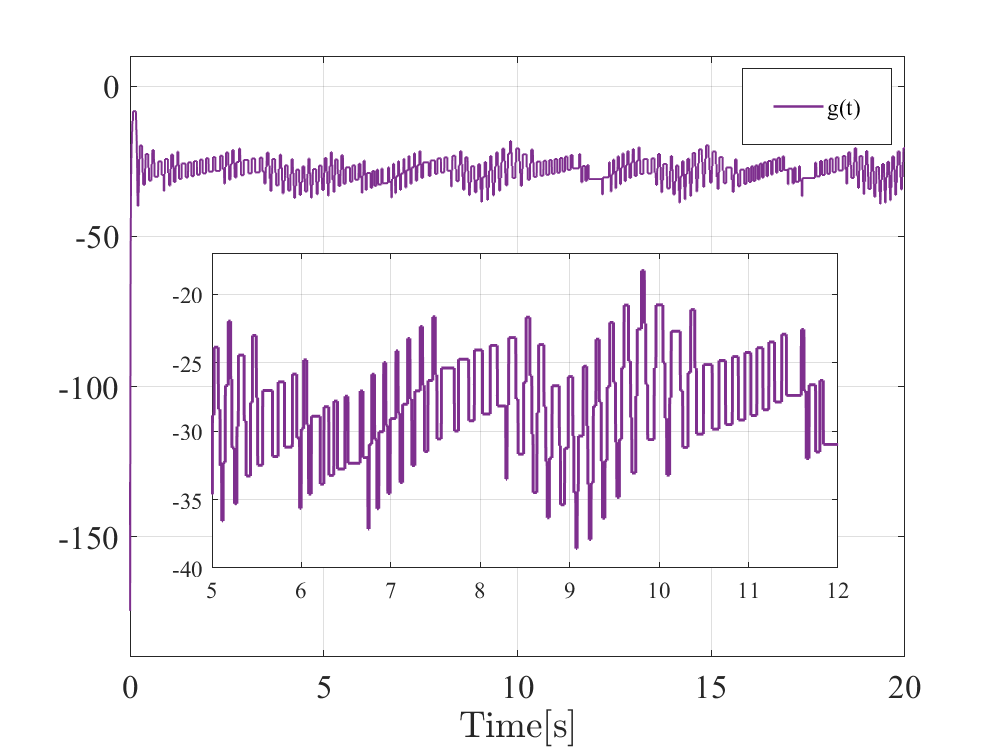}}
            \subfloat[ ]{\includegraphics[width=0.25\linewidth]{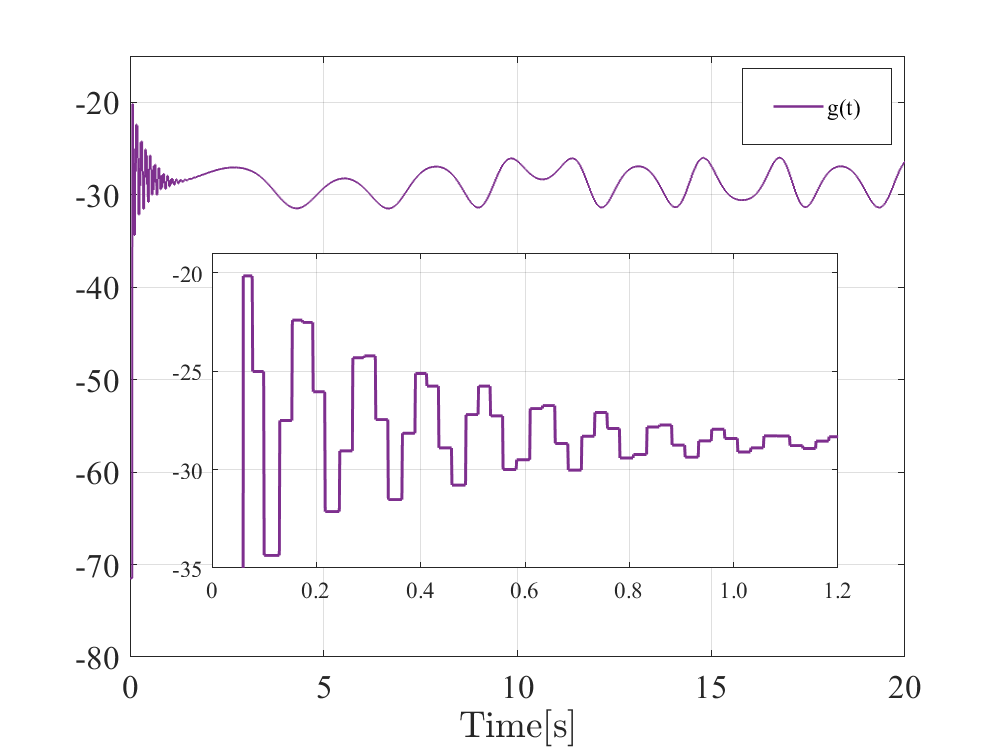}}
            \caption{The change of the actual controller $g$ under different ETC strategies, where (a), (b), (c), and (d) represent the fixed threshold, relative threshold, switched threshold and self-triggered strategies respectively.}
            \label{figure2}
        \end{figure*}

	\begin{figure*}
            \centering
            \subfloat[ ]{\includegraphics[width=0.25\linewidth]{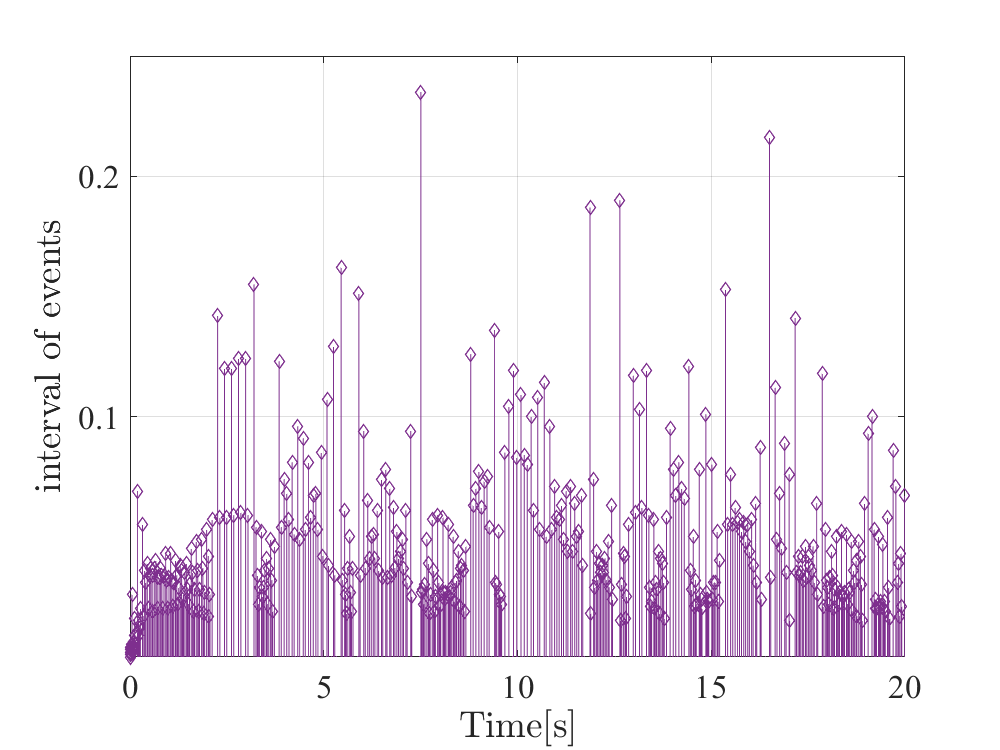}}
            \subfloat[ ]{\includegraphics[width=0.25\linewidth]{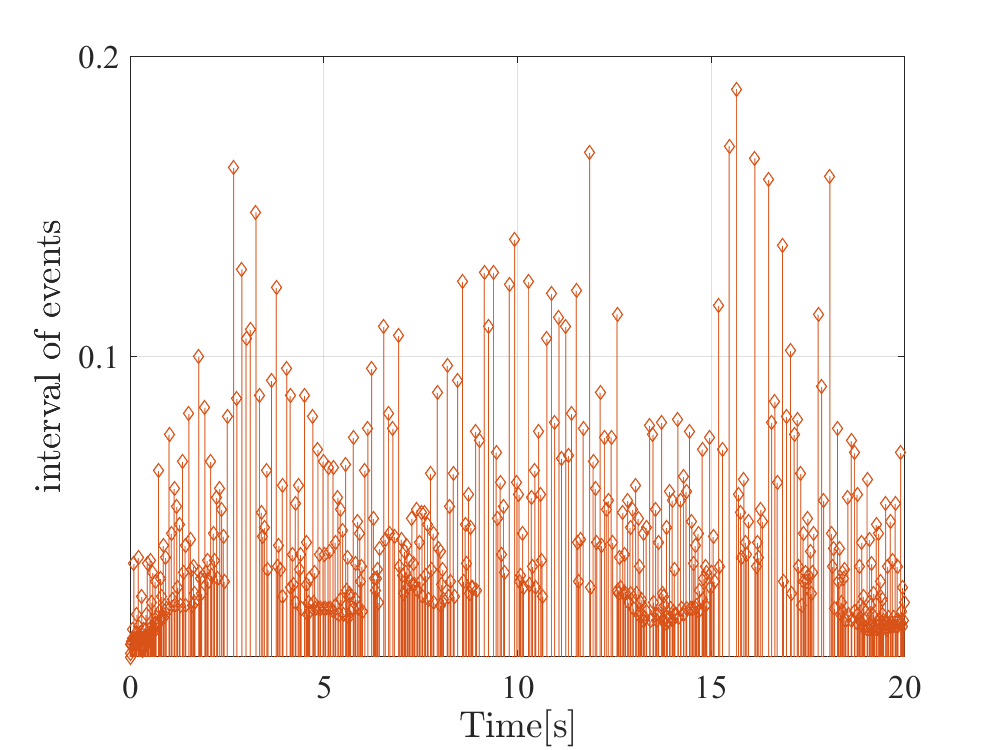}}
            \subfloat[ ]{\includegraphics[width=0.25\linewidth]{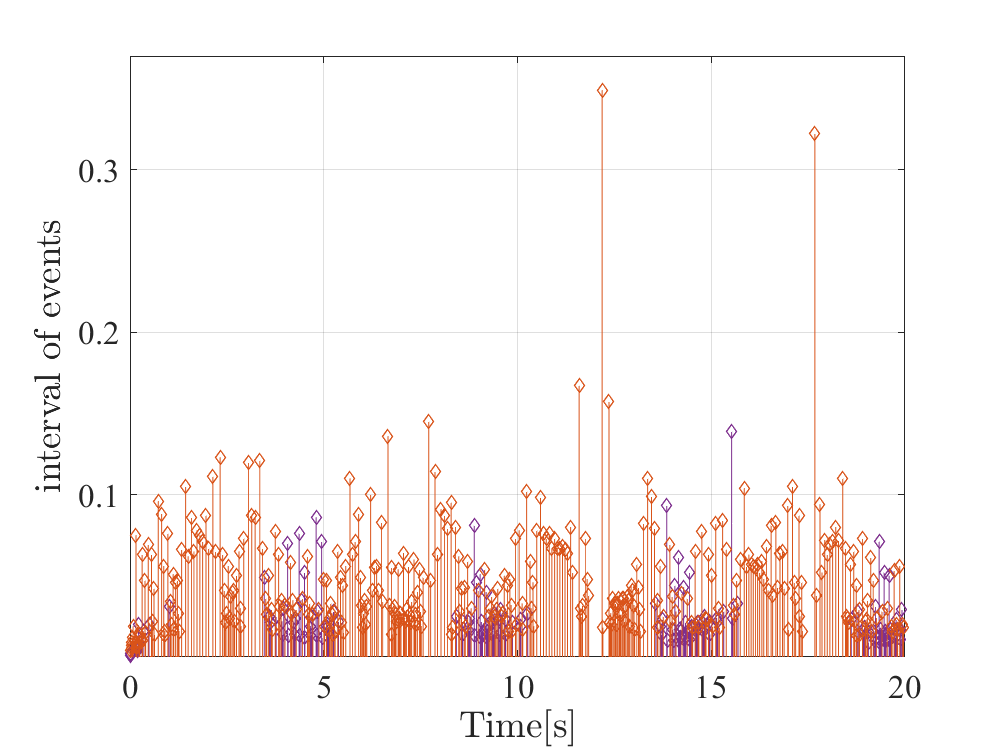}}
            \subfloat[ ]{\includegraphics[width=0.25\linewidth]{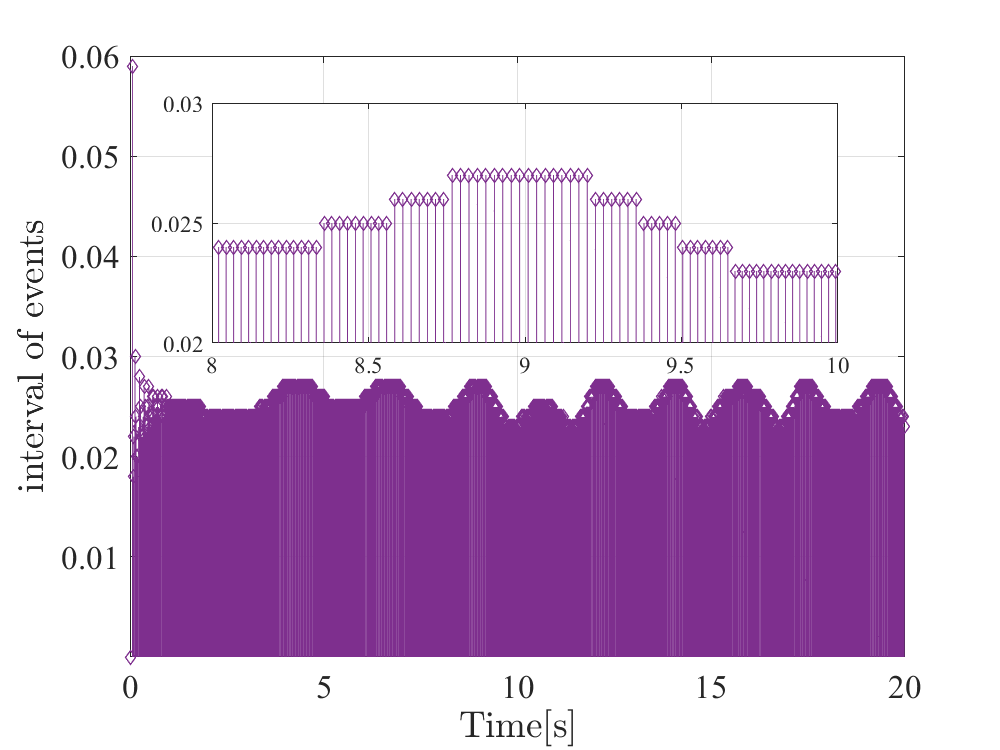}}
            \caption{The trigger interval for event triggering under different ETC strategies, where (a), (b), (c), and (d) represent the fixed threshold, relative threshold, switched threshold and self-triggered strategies respectively.}
            \label{figure3}
        \end{figure*}
        
        \begin{table}[htbp]
			\caption{Triggering counts of the different strategies\label{tab1}}
			\centering
			\begin{tabular}{|c|c|}
				\hline
				Fixed-threshold strategy & 439 \\
				\hline
				Relative-threshold strategy & 565 \\
				\hline
				Switched-threshold strategy & 151+347 \\
				\hline
				Self-triggered strategy & 798 \\
				\hline
			\end{tabular}
		\end{table}

		 Simulation results are shown in Figs.\ref{figure1}-\ref{figure3}. Fig.\ref{figure1} compares the target reference trajectory with the tracking output of the system under different ETC strategies. Trajectory error analysis demonstrates high-precision stable control, with tracking accuracy ranking: relative $>$ switched $>$ fixed $>$ self-triggered ETC strategy. The analysis shows that the more frequently the triggers occur, the better the tracking performance. Fig.\ref{figure2} reveals the changes of controllers. The controller remains unchanged when the triggering condition is not satisfied. Fig.\ref{figure3} illustrates the trigger interval distribution of multi-threshold event-triggered strategies. As shown in Table \ref{tab1}, trigger counts vary notably across strategies, yet all significantly reduce triggering frequency in 20,000 test cycles while maintaining tracking performance and minimizing computational costs.

		\section{Conclusion}\label{sec5}
		This paper integrates an adaptive fixed-time control algorithm with event-triggered control strategies to track fully state-constrained nonlinear systems. A state transformation using an asymmetric nonlinear mapping auxiliary system ensures closed-loop stability and tracking accuracy within preset durations. The proposed multi-threshold event-triggered strategies have been rigorously validated through comprehensive numerical simulations. The results demonstrate that the four types of triggers not only achieve excellent tracking accuracy but also significantly reduce trigger frequency, thereby conserving computational resources and optimizing overall system efficiency. In our furture research, we intend to broaden the existing framework to tackle optimal control issues for discrete event systems, like \cite{ji1,ji2}.
		 
         \vspace{-5pt}
		
		\appendix[proof of \eqref{eq111} and \eqref{eq11}]
		{\bf Step 1:} Construct the first Lyapunov function of the form: $
		V_{1}=\frac{1}{2}z_{1}^{2}+\frac{1}{2\varepsilon _{1}}\tilde{\varphi}_{1}^{2} $, 
		where $\varepsilon _{1}=2f_{1}/(2f_{1}-1)$ with $f_{1}>1/2$, and $\tilde{\varphi _{1}}=\varphi _{1}-\hat{\varphi } _{1} $ is the estimation error with $\hat{\varphi} _{1}$ being the estimation of $\varphi _{1}$ defined later. Then, derived by combining \eqref{eq7}-\eqref{eq10} with the above equation:
		\begin{eqnarray}\label{eq47}
			\dot{V}_{1}=z_{1}(-k_{2,1}z_{1}^{2p-1}+z_{2}+\alpha _{1}+U_{1}(Z_{1}))-\frac{\tilde{\varphi}_{1}\dot{\hat{\varphi}}_{1}}{\varepsilon}-\frac{z_{1}^{2}}{2}
		\end{eqnarray}
		where $U_{1}(Z_{1})=L_{1}(\bar{w}_{2})+\frac{1}{2}z_{1}+k_{2,1}z_{1}^{2p-1} $, and $k_{2,1}$ is a positive designed constant. $U_{1}(Z_{1})$ can be approximated by RBFNNs with accuracy $\lambda_{1}$, so that $U_{1}(Z_{1})=H_{1}^{T}\Omega _{1}(Z_{n})+\varpi (Z_{1})$, $\left \| \varpi (Z_{1}) \right \| \le \lambda _{1}$, where $\lambda_{1}>0$, $Z_{1}=[ w_{1},w_{s},\dot{w}_{s}]$. Based on the Young’s inequality, one obtains: $z_{1} U_{1}(Z_{1} ) \le (1/(2u_{1}^{2})) z_{1} ^{2} \varphi _{1} \Omega _{1} ^{T}(Z_{1} )\Omega _{1}(Z_{1} )+u_{1}^{2} /2+z_{1} ^{2}/2+\lambda _{1}^{2}/2$, where $\varphi_{1} =\left \| H_{1} \right \|^{2}$ are unknown constants to be estimated, and $u_{1}$ is a positive designed constant. Then, with positive designed constants $k_{1,1}$, $k_{2,1}$ and $\tau_{1}$, we design the virtual control law $\alpha_{1}$ and adaptive law $\dot{\hat{\varphi }}_{1}$ as
		\begin{eqnarray}
			\label{eq49}
			&\alpha_{1}=-k_{1,1} z_{1}^{2q-1}-\frac{1}{2u_{1}^{2}}z_{1} \hat{\varphi } _{1}  \Omega _{1} ^{T} (Z_{1} )\Omega _{1}(Z_{1})+\dot{w}_{s}\\
			\label{eq50}
			&\dot{\hat{\varphi}}_{1}=\frac{\varepsilon _{1}}{2u_{1}^{2}}z_{1}^{2}\Omega _{1}^{T}(Z_{n})\Omega _{1}(Z_{1})-\tau _{1}\hat{\varphi }_{1}
		\end{eqnarray}
		
		Bringing \eqref{eq49}-\eqref{eq50} into \eqref{eq47}, we can get the equation \eqref{eq111}.
		
		{\bf Step i $\bf(1<i<n)$:} Construct the ith Lyapunov function:$
		V_{i}=V_{i+1}+\frac{1}{2}z_{i}^{2}+\tilde{\varphi}_{i}^{2}/2\varepsilon _{i} $, where $\varepsilon _{i}=2f_{i}/(2f_{i}-1)$ with $f_{i}>1/2$. Then we have
		\begin{eqnarray}\label{eq53}
			\begin{split}
				\dot{V} _{i} =&-\sum_{o=1}^{i-1}k_{1,o}z_{o}^{2q} -\sum_{o=1}^{i-1}k_{2,o}z_{o}^{2p}+\sum_{o=1}^{i-1}(\frac{u_{o}^{2}}{2} +\frac{\lambda_{o}^{2}}{2})\\
				&+\sum_{o=1}^{i-1}\frac{\tau _{o}}{\varepsilon _{o}}\tilde{\varphi}_{o}\hat{\varphi} _{o}+z_{i}z_{i-1}-\frac{1}{2}z_{i}^{2}-\frac{1}{\varepsilon _{i} } \dot{\hat{\varphi}} _{i}\tilde{\varphi}_{i}\\
				&+z_{i}(-k_{2,i}z_{i}^{2p-1}+z_{i+1}+\alpha_{i}+U_{i}(Z_{i}))
			\end{split}
		\end{eqnarray}
        
		Similar to \eqref{eq47}, $U_{i}(z_{i})$ is equal to $
		U_{i}(Z_{i})=z_{i-1}+L_{i}(\bar{w}_{i+1})-\dot{\alpha}_{i-1}+\frac{1}{2}z_{i}+K_{2,i}z_{i}^{2p-1} $, 
		where $\dot{\alpha}_{i-1}=\sum_{o=1}^{i-1}\frac{\partial \alpha _{i-1} }{\partial w_{o}} ( w_{o+1}+L_{s}(\bar{w}_{o+1}))+\sum_{o=1}^{i-1} \frac{\partial \alpha _{i-1} }{\partial \hat{\varphi }_{o}}\dot{\hat{\varphi}}_{o}+\sum_{o=0}^{i-1}\frac{\partial \alpha _{i-1} }{\partial w_{s}^{o}}w_{s}^{o+1}$. Then, with positive designed constants $k_{1,i}$, $k_{2,i}$ and $\tau_{i}$, we design the virtual control law $\alpha_{i}$ and adaptive law $\dot{\hat{\varphi }}_{i}$ as
		\begin{eqnarray}\label{eq55}
			&\alpha_{i}=-k_{1,i} z_{i}^{2q-1}-\frac{1}{2u_{i}^{2}}z_{i} \hat{\varphi } _{i}  \Omega _{i} ^{T} (Z_{i} )\Omega _{i}(Z_{i})\\
			\label{eq56}
			&\dot{\hat{\varphi}}_{i}=\frac{\varepsilon _{i}}{2u_{i}^{2}}z_{i}^{2}\Omega _{i}^{T}(Z_{n})\Omega _{i}(Z_{i})-\tau _{i}\hat{\varphi }_{i}
		\end{eqnarray}
		
		Bringing \eqref{eq55}-\eqref{eq56} into \eqref{eq53}, we can get the equation \eqref{eq11}.

	\end{document}